\documentclass[12pt,a4paper]{amsart}

\usepackage{amsmath,amsfonts,amssymb}
\usepackage{xcolor}
\usepackage[hmargin=2cm,vmargin=2cm]{geometry}

\usepackage[hypertexnames=false,hyperfootnotes=false,colorlinks=true,linkcolor=blue,%
citecolor=purple,filecolor=magenta,urlcolor=cyan,unicode,linktocpage=true,pagebackref=false]{hyperref}
\usepackage{nameref,zref-xr}     

\usepackage[all]{xy}   
\usepackage{tikz}

\usepackage{caption}
\usepackage{subcaption}
\usepackage{array}

\usepackage{enumerate}

\newcommand{\mbP}{\mathbb P}
\newcommand{\mbZ}{\mathbb Z}
\newcommand{\mbC}{\mathbb C}

\newcommand{\oM}{\overline{\mathcal M}}

\newcommand{\tu}{{\widetilde u}}

\def\oM{{\overline{\mathcal{M}}}}
\def\CP{{{\mathbb C}{\mathbb P}}}

\def\d{{\partial}}

\newcommand{\eps}{\varepsilon}

\newcommand{\cA}{\mathcal A}
\newcommand{\hcA}{\widehat{\mathcal A}}
\newcommand{\DR}{\mathrm{DR}}

\newcommand{\even}{\mathrm{even}}

\newcommand{\Coef}{\mathrm{Coef}}

\newcommand{\ST}{\mathrm{ST}}
\newcommand{\tv}{\widetilde v}

\newcommand{\gl}{\mathrm{gl}}

\newcommand{\tT}{\widetilde{T}}

\def\d{\partial}

\newcommand{\beq}{\begin{equation}}
\newcommand{\eeq}{\end{equation}}

\newcommand{\Id}{\mathrm{Id}}

\newcommand{\End}{\mathrm{End}}

\newcommand{\triv}{\mathrm{triv}}

\newcommand{\KdV}{\mathrm{KdV}}

\newtheorem{theorem}{Theorem}[section]

\newtheorem{lemma}[theorem]{Lemma}

\newtheorem{example}[theorem]{Example}
\newtheorem{remark}[theorem]{Remark}

\title[Integrable systems of finite type from F-CohFTs without unit]{Integrable systems of finite type from F-cohomological field theories without unit}

\author{Alexandr Buryak}
\address[A. Buryak]{Faculty of Mathematics, National Research University Higher School of Economics, 6 Usacheva str., Moscow, 119048, Russian Federation;\smallskip\newline 
Center for Advanced Studies, Skolkovo Institute of Science and Technology, 1 Nobel str., Moscow, 143026, Russian Federation;\smallskip\newline
P.G. Demidov Yaroslavl State University, 14 Sovetskaya str., Yaroslavl, 150003, Russian Federation}
\email{aburyak@hse.ru}

\author{Danil Gubarevich}
\address[D. Gubarevich]{Faculty of Mathematics, National Research University Higher School of Economics, 6 Usacheva str., Moscow, 119048, Russian Federation}
\email{danilphys180916@mail.ru}

\numberwithin{equation}{section}

\renewcommand{\gg}[2]{\fill[color=white] (#2) circle(3mm) node {\color{black}$#1$}; \draw (#2) circle (3mm)}
\newcommand{\lab}[4]{\draw (#1)++(#2:#3) node {$#4$};}
\newcommand{\leg}[2]{\begin{scope}[shift={(#1)}] \draw (0:0) -- (#2:8mm);\end{scope}}

\begin{document}

\tikz{\coordinate (A) at (0,0);\coordinate (B) at (0,14mm);\coordinate (BL) at (-12mm,14mm);\coordinate (BR) at (12mm,14mm);

}

\begin{abstract}
One of many manifestations of a deep relation between the topology of the moduli spaces of algebraic curves and the theory of integrable systems is a recent construction of Arsie, Lorenzoni, Rossi, and the first author associating an integrable system of evolutionary PDEs to an F-cohomological field theory (F-CohFT), which is a collection of cohomology classes on the moduli spaces of curves satisfying certain natural splitting properties. Typically, these PDEs have an infinite expansion in the dispersive parameter, which happens because they involve contributions from the moduli spaces of curves of arbitrarily large genus. In this paper, for each rank $N\ge 2$, we present a family of F-CohFTs without unit, for which the equations of the associated integrable system have a finite expansion in the dispersive parameter. For $N=2$, we explicitly compute the primary flows of this integrable system.
\end{abstract}

\date{\today}

\maketitle

\section{Introduction}

The fact that integrable systems provide an appropriate tool for the description of the topology of the moduli spaces~$\oM_{g,n}$ of stable algebraic curves of genus $g$ with $n$ marked points was first observed by Witten~\cite{Wit91} in his famous conjecture, proved by Kontsevich~\cite{Kon92}. This result says that the generating series of integrals over $\oM_{g,n}$ of monomials in psi-classes (the first Chern classes of tautological line bundles) is controlled by a special solution of the Korteweg--de Vries (KdV) hierarchy. Various versions of Witten's conjecture were proposed (see, e.g.,~\cite{DZ04,OP06,Wit93,FSZ10}), when it was realized that integrable systems appear in a very general context, where the central role is played by the notion of a \emph{cohomological field theory} (CohFT) introduced by Kontsevich and Manin~\cite{KM94}. CohFTs are systems of cohomology classes on the moduli spaces~$\oM_{g,n}$ that are compatible with natural morphisms between the moduli spaces. There is a general result~\cite{BPS12} saying that the generating series of correlators of an arbitrary semisimple CohFT is controlled by a special solution of a certain integrable system, determined by the CohFT uniquely. This integrable system is called the \emph{hierarchy of topological type} or the \emph{Dubrovin--Zhang (DZ) hierarchy}.

\medskip

\begin{remark}
Since the concept of an integrable system can be interpreted in various ways, we would like to clarify that for the purposes of this paper we define an integrable system as an infinite collection of pairwise commuting evolutionary flows.
\end{remark}

\medskip

It was first observed by Dubrovin and Zhang~\cite{DZ01} that the hierarchies of topological type form a wide class of integrable systems that can be conjecturally described independently of the geometry, using only the language of integrable systems (see further developments in~\cite{DLYZ16,LWZ21}). Thus, the topology of the moduli spaces of algebraic curves can be viewed as a tool for constructing integrable systems.

\medskip

There is another construction of an integrable system, the so-called \emph{double ramification (DR) hierarchy}, associated to a CohFT, suggested in~\cite{Bur15}. For a semisimple CohFT, the DR hierarchy is conjecturally Miura equivalent to the DZ hierarchy. The advantage of the DR hierarchy is that its equations are constructed very explicitly in terms of certain integrals over~$\oM_{g,n}$. Moreover, the DR hierarchy is defined for objects that are more general than CohFTs: for the so-called \emph{F-cohomological field theories} (F-CohFTs)~\cite{ABLR21,ABLR23}.

\medskip

Typically, the equations of the DR hierarchy associated to an F-CohFT have an infinite expansion in the dispersive parameter. This happens because the equations involve contributions from $\oM_{g,n}$ with arbitrarily large $g$. So it is natural to try to determine F-CohFTs, for which the associated DR hierarchy is of \emph{finite type}, meaning that all the equations have a finite expansion in the dispersive parameter. For CohFTs associated to simple singularities, the DR hierarchy is of finite type. For the partial CohFTs associated to these CohFTs~\cite{LRZ15}, the DR hierarchy is also of finite type. However, as far as we know, no other examples of F-CohFTs for which the associated DR hierarchy is of finite type have been found.

\medskip

In this paper, for each $N\ge 2$, we present a family of F-CohFTs without unit of rank $N$, depending on a vector $G\in\mbC^N$ and a strictly upper triangular $N\times N$ matrix $R_1$ satisfying $R_1^2=0$, and prove that the associated DR hierarchy is of finite type. Finally, for $N=2$, we compute explicitly the primary flows $\frac{\d}{\d t^1_0}$ and $\frac{\d}{\d t^2_0}$ of the DR hierarchy.

\medskip

\subsection*{Notation and conventions}
 
\begin{itemize}
\item Throughout the text we use the Einstein summation convention for repeated upper and lower Greek indices.

\smallskip

\item When it doesn't lead to a confusion, we use the symbol $*$ to indicate any value, in the appropriate range, of a sub- or superscript.

\smallskip

\item For a topological space~$X$ let $H^*(X)$ denote the cohomology ring of~$X$ with coefficients in~$\mbC$.
\end{itemize}

\medskip

\subsection*{Acknowledgements}

The work of A.~B. is supported by the Russian Science Foundation (Grant no. 20-71-10110). A. B. is grateful to A. Mikhailov, P. Rossi, and V. Sokolov for motivating discussions about the finiteness of the integrable systems associated to F-CohFTs. 

\medskip


\section{F-cohomological field theories without unit and DR hierarchies}

\subsection{F-cohomological field theories without unit}

An \emph{F-cohomological field theory without unit} (F-CohFT without unit) is a system of linear maps 
$$
c_{g,n+1}\colon V^*\otimes V^{\otimes n} \to H^\even(\oM_{g,n+1}),\quad 2g-1+n>0,
$$
where $V$ is an arbitrary finite dimensional vector space, such that the following axioms are satisfied.
\begin{enumerate}[(i)]
\item The maps $c_{g,n+1}$ are equivariant with respect to the $S_n$-action permuting the $n$ copies of~$V$ in $V^*\otimes V^{\otimes n}$ and the last $n$ marked points in $\oM_{g,n+1}$, respectively.

\smallskip

\item Fixing a basis $e_1,\ldots,e_{\dim V}$ in $V$ and the dual basis $e^1,\ldots,e^{\dim V}$ in $V^*$, the following property holds:
$$
\gl^* c_{g_1+g_2,n_1+n_2+1}(e^{\alpha_0}\otimes\otimes_{i=1}^{n_1+n_2} e_{\alpha_i}) = c_{g_1,n_1+2}(e^{\alpha_0}\otimes \otimes_{i\in I} e_{\alpha_i} \otimes e_\mu)\otimes c_{g_2,n_2+1}(e^{\mu}\otimes \otimes_{j\in J} e_{\alpha_j})
$$
for $1 \leq\alpha_0,\alpha_1,\ldots,\alpha_{n_1+n_2}\leq \dim V$, where $I \sqcup J = \{2,\ldots,n_1+n_2+1\}$, $|I|=n_1$, $|J|=n_2$, and $\gl\colon\oM_{g_1,n_1+2}\times\oM_{g_2,n_2+1}\to \oM_{g_1+g_2,n_1+n_2+1}$ is the corresponding gluing map. Clearly the axiom doesn't depend on the choice of a basis in $V$.
\end{enumerate}
The dimension of $V$ is called the \emph{rank} of the F-CohFT without unit.

\medskip

An \emph{F-CohFT} is an F-CohFT without unit endowed with a nonzero vector $e\in V$, called the \emph{unit}, such that the following additional property is satisfied:
\begin{enumerate}[(i)]
\setcounter{enumi}{2}
\item $\pi^* c_{g,n+1}(\omega\otimes\otimes_{i=1}^n v_i)=c_{g,n+2}(\omega\otimes\otimes_{i=1}^n  v_i\otimes e)$ for $\omega\in V^*$ and $v_1,\ldots,v_n\in V$, where $\pi\colon\oM_{g,n+2}\to\oM_{g,n+1}$ is the map that forgets the last marked point. Moreover, $c_{0,3}(\omega\otimes v \otimes e) = \omega(v)$ for $\omega\in V^*$ and $v\in V$.
\end{enumerate}

\medskip

An F-CohFT is called an \emph{F-topological field theory} (F-TFT) if  $c_{g,n+1}(\omega\otimes\otimes_{i=1}^n v_i)\in H^0(\oM_{g,n+1})$ for all $\omega\in V^*$ and $v_1,\ldots,v_n\in V$.

\medskip

\subsection{DR hierarchy}

Let us fix $N\ge 1$ and let $u^1,\ldots,u^N$ be formal variables. To the formal variables $u^\alpha$ we attach formal variables $u^\alpha_d$ with $d\ge 0$ and introduce the ring of \emph{differential polynomials} $\cA:=\mbC[[u^*]][u^*_{\ge 1}]$. We identify $u^\alpha_0=u^\alpha$ and also denote $u^\alpha_x:=u^\alpha_1$, $u^{\alpha}_{xx}:=u^\alpha_2$, \ldots. An operator $\d_x\colon\cA\to\cA$ is defined by $\d_x:=\sum_{n\ge 0}u^\alpha_{n+1}\frac{\d}{\d u^\alpha_n}$. The extended space of differential polynomials is defined by $\hcA:=\cA[[\eps]]$. 

\medskip

Denote by $\psi_i\in H^2(\oM_{g,n})$ the first Chern class of the line bundle over~$\oM_{g,n}$ formed by the cotangent lines at the $i$-th marked point of stable curves. Denote by~$\mathbb E$ the rank~$g$ Hodge vector bundle over~$\oM_{g,n}$ whose fibers are the spaces of holomorphic one-forms on stable curves. Let $\lambda_j:= c_j(\mathbb E)\in H^{2j}(\oM_{g,n})$. 

\medskip

For any $a_1,\dots,a_n\in \mbZ$, $\sum_{i=1}^n a_i =0$, let $\DR_g(a_1,\ldots,a_n) \in H^{2g}(\oM_{g,n})$ be the {\it double ramification (DR) cycle}. The DR cycle is the pushforward, through the forgetful map to $\oM_{g,n}$, of the virtual fundamental class of the moduli space of projectivized stable maps to~$\CP^1$ relative to~$0$ and~$\infty$, with ramification profile $a_1,\ldots,a_n$ at the marked points (see, e.g.,~\cite{BSSZ15} for more details). More precisely, the pushforward itself lies in $H_{2(2g-3+n)}(\oM_{g,n})$, while its Poincar\'e dual cohomology class lies in $H^{2g}(\oM_{g,n})$. By abuse of notation, we will denote both the pushforward and its Poincar\'e dual by $\DR_g(a_1,\ldots,a_n)$. The crucial property of the DR cycle is that for any cohomology class $\theta\in H^*(\oM_{g,n})$ the integral $\int_{\oM_{g,n+1}}\lambda_g\DR_g\left(-\sum a_i,a_1,\ldots,a_n\right)\theta$ is a homogeneous polynomial in $a_1,\ldots,a_n$ of degree~$2g$ (see, e.g.,~\cite{Bur15}). 

\medskip

Consider now an arbitrary F-CohFT without unit of rank $N$ and define differential polynomials $P^\alpha_{\beta,d}\in\hcA$, $1\le\alpha,\beta\le N$, $d\ge 0$, by
\begin{equation*}
P^\alpha_{\beta,d}:=\sum_{\substack{g,n\geq 0,\,2g+n>0\\k_1,\ldots,k_n\geq 0\\\sum_{j=1}^n k_j=2g}} \frac{\eps^{2g}}{n!} \Coef_{(a_1)^{k_1}\ldots(a_n)^{k_n}} \left(\int_{\DR_g(-\sum_{j=1}^n a_j,0,a_1,\ldots,a_n)} \hspace{-2.3cm}\lambda_g \psi_2^d c_{g,n+2}(e^\alpha\otimes e_\beta\otimes \otimes_{j=1}^n e_{\alpha_j}) \right)\prod_{j=1}^n u^{\alpha_j}_{k_j}.
\end{equation*}
The \emph{DR hierarchy} is the following system of evolutionary PDEs:
\begin{gather}\label{eq:DR hierarchy}
\frac{\d u^\alpha}{\d t^\beta_d}=\d_x P^\alpha_{\beta,d},\qquad 1\le \alpha,\beta\le N,\quad d\ge 0.
\end{gather}

\medskip

\begin{theorem}
All the equations of the DR hierarchy are compatible with each other, namely, 
$$
\frac{\d}{\d t^{\beta_2}_{d_2}}\left(\frac{\d u^\alpha}{\d t^{\beta_1}_{d_1}}\right) = \frac{\d}{\d t^{\beta_1}_{d_1}}\left(\frac{\d u^\alpha}{\d t^{\beta_2}_{d_2}}\right),\qquad 1\leq \alpha,\beta_1,\beta_2\leq N,\quad d_1,d_2\geq 0.
$$
\end{theorem}
\begin{proof}
For F-CohFTs, the theorem is proved in~\cite[Theorem~5.1]{BR21} (see more details in~\cite[Theorem~2]{ABLR21}), however, the existence of a unit is never used there. So the same proof works for F-CohFTs without unit.
\end{proof}

\medskip

\begin{example}\label{example:KdV}
Consider the trivial F-CohFT given by $V=\mbC$, $e_1=e=1\in\mbC=V$, and
$$
c^\triv_{g,n+1}(e^1\otimes e_1^{\otimes n}):=1\in H^0(\oM_{g,n+1}).
$$
Then the corresponding DR hierarchy is the KdV hierarchy~\cite[Section~4.3.1]{Bur15} (we denote $u_d:=u^1_d$ and $t_d:=t^1_d$)
$$
\frac{\d u}{\d t_d}=\d_x P^\KdV_d,
$$
where 
$$
P^\KdV_0=u,\quad P^\KdV_1=\frac{u^2}{2}+\frac{\eps^2}{12}u_{xx},\quad P_2^\KdV=\frac{u^3}{6}+\frac{\eps^2}{24}(2u u_{xx}+u_x^2)+\frac{\eps^4}{240}u_{xxxx},
$$
and a general formula for $P^\KdV_d$ is
$$
\d_x P^\KdV_d=\frac{\eps^{2d+2}}{2(2d+1)!!}\left[\left(L^{d+\frac{1}{2}}\right)_+,L\right],\quad L=\d_x^2+2\eps^{-2}u.
$$
\end{example}

\medskip


\section{A family of F-CohFTs without unit and the finiteness of the DR hierarchy}

\subsection{$R$-matrices}

Let us briefly recall a group action on F-CohFTs without unit constructed in~\cite[Section~4.1]{ABLR23}.

\medskip

Let us fix a finite dimensional vector space $V$ of dimension $N$ and consider the group $G_+$ of $\End(V)$-valued formal power series of the form $R(z)=\Id+\sum_{i\geq 1} R_iz^i$. Let us denote by $R^{-1}(z)$ the inverse element to $R(z)$ and by $R(z)^t$ the transposed $\End(V^*)$-valued power series. We refer to such an element of $G_+$ as an \emph{$R$-matrix}.

\medskip

Let $\Gamma$ be a stable graph of genus $g$ with $n$ marked legs (see \cite[Section~0.2]{PPZ15} for the definition) and $V(\Gamma)$, $E(\Gamma)$ be its sets of vertices and edges, each vertex $v\in V(\Gamma)$ marked with a genus $g(v)$ and with valence $n(v)$. We denote by $E[v]$ the set of edges incident to $v$. Let $\xi_{\Gamma}\colon\prod_{v\in V(\Gamma)} \oM_{g(v),n(v)} \to \oM_{g,n}$ be the natural map whose image is the closure of the locus of stable curves whose dual graph is $\Gamma$. 

\medskip

By \emph{stable tree} we mean a stable graph $\Gamma$ with the first Betti number $b_1(\Gamma)$ equal to zero. Let~$\ST_{g,n+1}$ be the set of stable trees of genus $g$ with $n+1$ marked legs. Then $T \in\ST_{g,n+1}$ can be seen as a rooted tree where the root is the vertex to which leg~$1$ is attached and each edge $e\in E(T)$ is splitted into two half-edges $e'$ and $e''$, where~$e'$ is closer to the root and~$e''$ is farther from the root. Consider a function $l_T\colon V(T)\to\mbZ_{\ge 1}$ that is uniquely determined by the condition that its value on the root is equal to $1$ and that if a vertex $v$ is the mother of a vertex $v'$, then $l_T(v')=l_T(v)+1$. The number $l_T(v)$ is called the \emph{level} of a vertex $v\in V(T)$. The number $\deg(T):=\max_{v\in V(T)}l_T(v)$ is called the \emph{degree} of $T$. 

\medskip

The action of $R(z)\in G_+$ on an F-CohFT without unit $c_{g,n+1}\colon V^*\otimes V^{\otimes n} \to H^\even(\oM_{g,n+1})$ is the system of maps
\begin{equation}\label{eq:R-action}
(R(z)c)_{g,n+1} := \sum_{T \in\ST_{g,n+1}}\xi_{T*}\left[\prod_{v\in V(T)}c_{g(v),n(v)}R(-\psi_1)^t \prod_{k=2}^{n+1} R^{-1}(\psi_k)\prod_{e\in E(T)} \frac{\Id-R^{-1}(\psi_{e'})R(-\psi_{e''})}{\psi_{e'}+\psi_{e''}}\right].
\end{equation}
In~\cite[Theorem~4.3]{ABLR23} the authors proved that $(R(z)c)_{g,n+1}$ is again an F-CohFT without unit.

\medskip

\subsection{DR hierarchies of finite type}\label{subsection:DR hierarchies of finite type}

Let $N\ge 2$, $V=\mbC^N$, and $e_1,\ldots,e_N\in\mbC^N$ be the standard basis in $\mbC^N$. Following~\cite[Section~4.4]{ABLR23}, consider the following F-TFT, parameterized by a vector $G=(G^1,\ldots,G^N)\in\mbC^N$:
$$
c^{\triv,G}_{g,n+1}(e^{i_0}\otimes\otimes_{j=1}^n e_{i_j}):=
\begin{cases}
(G^{i_0})^g,&\text{if $i_0=i_1=\ldots=i_n$},\\
0,&\text{otherwise}.
\end{cases}
$$

\medskip

\begin{theorem}\label{theorem:finite DR hierarchy}
Consider an arbitrary vector $G=(G^1,\ldots,G^N)\in\mbC^N$ and a strictly upper triangular $N\times N$ matrix $R_1$ such that $R_1^2=0$. Then the DR hierarchy corresponding to the F-CohFT without unit $\left((\Id+R_1 z)c^{\triv,G}\right)_{g,n+1}$ is of finite type, i.e., all the differential polynomials~$P^\alpha_{\beta,d}$ from~\eqref{eq:DR hierarchy} are polynomials in $\eps$.
\end{theorem}
\begin{proof}
Since $(\Id+R_1 z)^{-1}=\Id-R_1 z$, we have
\begin{gather}\label{eq:sum of trees for our family}
\left((\Id+R_1 z)c^{\triv,G}\right)_{g,n+1}=\sum_{T \in \ST_{g,n+1}}\xi_{T*}\left[\prod_{v\in V(T)}c^{\triv,G}_{g(v),n(v)}(\Id-R_1^t\psi_1)\prod_{k=2}^{n+1}(\Id-R_1\psi_k)\prod_{e\in E(T)}R_1\right].
\end{gather}
Note that in this formula each $c_{g(v),n(v)}$ is fed by a covector $e^{\alpha_1(v)}$ and vectors $e_{\alpha_2(v)},\ldots,e_{\alpha_{n(v)}(v)}$, and the result is zero unless $\alpha_1(v)=\ldots=\alpha_{n(v)}(v)$. Suppose that this equality is satisfied and denote $\alpha(v):=\alpha_i(v)$. Since the matrix $R_1$ is strictly upper triangular, we see that the result is zero unless $\alpha(v')>\alpha(v)$ if a vertex $v$ is the mother of a vertex $v'$. We conclude that only stable trees $T$ with $\deg(T)\le N$ can give a nontrivial contribution on the right-hand side of~\eqref{eq:sum of trees for our family}.

\medskip

Consider the DR hierarchy corresponding to the F-CohFT without unit $\left((\Id+R_1 z)c^{\triv,G}\right)_{g,n+1}$. The coefficients of a differential polynomial $P^\alpha_{\beta,d}$ are determined by the integrals 
$$
\int_{\DR_g(-\sum_{j=1}^n a_j,0,a_1,\ldots,a_n)}\lambda_g \psi_2^d \left((\Id+R_1 z)c^{\triv,G}\right)_{g,n+2}(e^\alpha\otimes e_\beta\otimes \otimes_{j=1}^n e_{\alpha_j}).
$$
The contribution of a stable tree $T\in\ST_{g,n+2}$ is given by 
$$
\int_{\DR_g(-\sum_{j=1}^n a_j,0,a_1,\ldots,a_n)}\lambda_g \psi_2^d\xi_{T*}\left[\prod_{v\in V(T)}c^{\triv,G}_{g(v),n(v)}(\Id-R_1^t\psi_1)\prod_{k=2}^{n+1}(\Id-R_1\psi_k)\prod_{e\in E(T)}R_1\right].
$$
Denote by $\tv$ the vertex of $T$ incident to the leg number $2$. Using a splitting property of the DR cycle (see, e.g.,~\cite[Proposition~4.6]{BR22}) and of the class $\lambda_g$, we see that this integral is equal to a linear combination of the products $\prod_{v\in V(T)}I_T(v)$ of integrals $I_T(v)$ of the form
$$
I_T(v)=
\begin{cases}
\int_{\DR_{g(v)}(b_1,\ldots,b_{n(v)})}\lambda_{g(v)}\prod_{i=1}^{n(v)}\psi_i^{l_i},&\text{if $v\ne\tv$},\\
\int_{\DR_{g(v)}(0,b_2,\ldots,b_{n(v)})}\lambda_{g(v)}\psi_1^d\prod_{i=1}^{n(v)}\psi_i^{l_i},&\text{if $v=\tv$},
\end{cases}
$$
where $0\le l_i\le 1$. Degree counting immediately gives that $I_T(v)=0$ unless $g(v)+\sum_{i=1}^{n(v)}(1-l_i)=3+\delta_{v,\tv}d$. This implies that $I_T(v)=0$ unless $g(v)+|E[v]|\le3+\delta_{v,\tv}d$.

\medskip

\begin{lemma}
Let $N\ge 2$. Consider a rooted tree $T$ (without half-edges) endowed with a function $g\colon V(T)\to\mbZ_{\ge 0}$ and a chosen vertex $\tv\in V(T)$ such that $\deg(T)\le N$ and for any $v\in V(T)$ we have $g(v)+|E[v]|\le 3+\delta_{v,\tv}d$. Then $\sum_{v\in V(T)}g(v)\le (d+3)2^{N-1}$.
\end{lemma}
\begin{proof}
Elementary combinatorial considerations show that a unique tree $\tT$ satisfying the conditions from the statement of the lemma and having the maximal total genus $\sum_{v\in V(T)}g(v)$ can be described as follows:
\begin{itemize}
\item $\deg(\tT)=N$;

\smallskip

\item the root of $\tT$ is of genus $0$ and has exactly $3+d$ direct descendents;

\smallskip

\item each vertex $v$ with $2\le l_{\tT}(v)\le N-1$ is of genus $0$ and has exactly $2$ direct descendents;

\smallskip

\item each vertex $v$ with $l_{\tT}(v)=N$ is of genus $2$;
\end{itemize}
and we have $\sum_{v\in V(\tT)}g(v)=(d+3)2^{N-1}$. Indeed, if a tree $T$ has $\deg(T)\le N-1$, then we can attach a new vertex of genus $2$ to a vertex of $T$ of maximal level, probably having to decrease its genus by $1$, but the total genus increases after that. Then, if $\deg(T)=N$ and there is a vertex $v\in V(T)$ with $l_T(v)\le N-1$ and $g(v)\ge 1$, then one can increase the total genus by decreasing the genus of $v$ by one and attaching a new vertex of genus $2$ to it. This shows that one should look for a tree with the maximal total genus only among the trees $T$ with $\deg(T)=N$ and such that $g(v)=0$ if $l_T(v)\le N-1$, and $g(v)=2$ if $l_T(v)=N$. Clearly, $\tT$ has the maximal number of vertices of level $N$ among such trees.
\end{proof}

\medskip

This lemma clearly completes the proof of the theorem.
\end{proof}

\medskip


\section{The primary flows in the case of rank $2$}

The goal of this section is to compute explicitly the primary flows $\frac{\d}{\d t^\alpha_0}$ of the DR hierarchy from Theorem~\ref{theorem:finite DR hierarchy} in the case $N=2$. So the matrix $R_1$ has the form
\begin{gather}\label{eq:R1 for N2}
R_1=
\begin{pmatrix}
0 & \xi \\ 
0 & 0
\end{pmatrix},\quad\xi\in\mbC.
\end{gather}

\medskip

\begin{theorem}
For $G=(G^1,G^2)\in\mbC^2$ and $R_1$ given by~\eqref{eq:R1 for N2}, consider the F-CohFT without unit $\left((\Id+R_1 z)c^{\triv,G}\right)_{g,n+1}$ and the corresponding DR hierarchy.
\begin{enumerate}
\item[1.] After the Miura transformation 
\begin{gather}\label{eq:special Miura transformation}
\tu^1=u^1+\xi\frac{(u^2)^2}{2}+\frac{\eps^2}{24}\d_x^2\left(\xi G^2 u^2+\frac{G^1}{1+\xi u^2}\right),\qquad \tu^2=u^2,
\end{gather}
the flows $\frac{\d}{\d t^1_0}$ and $\frac{\d}{\d t^2_0}$ of the DR hierarchy become
\begin{align*}
\frac{\d\tu^1}{\d t^1_0}=&\d_x\left[\frac{\tu^1}{1+\xi \tu^2}\right],\\
\frac{\d \tu^2}{\d t^1_0}=&0,\\
\frac{\d\tu^1}{\d t^2_0}=&\xi\d_x\left[\frac{\tu^1\tu^2}{1+\xi\tu^2}-\frac{1}{2}\frac{(\tu^1)^2}{(1+\xi\tu^2)^2}-\frac{\eps^2 G^1}{12}\left(\left(\left(\frac{\tu^1}{1+\xi\tu^2}\right)_x\frac{1}{1+\xi\tu^2}\right)_x\frac{1}{1+\xi\tu^2}\right)\right],\\
\frac{\d\tu^2}{\d t^2_0}=&\tu^2_x.
\end{align*}

\smallskip

\item[2.] Moreover, we have $\frac{\d\tu^2}{\d t^1_d}=0$ and $\frac{\d\tu^2}{\d t^2_d}=\left.\d_x P_d^\KdV\right|_{u_n=\tu^2_n,\,\eps\mapsto\sqrt{G^2}\eps}$.
\end{enumerate}
\end{theorem}
\begin{proof}
The differential polynomials $P^\alpha_{\beta,d}$ of our DR hierarchy are given by the integrals
\begin{gather}\label{eq:DR integral for N2}
\int_{\DR_g(-\sum_{j=1}^n a_j,0,a_1,\ldots,a_n)}\lambda_g\psi_2^d\left((\Id+R_1 z)c^{\triv,G}\right)_{g,n+2}(e^\alpha\otimes e_\beta\otimes \otimes_{j=1}^n e_{\alpha_j}).
\end{gather}
Using formula~\eqref{eq:sum of trees for our family}, we express the class $\left((\Id+R_1 z)c^{\triv,G}\right)_{g,n+2}(e^\alpha\otimes e_\beta\otimes \otimes_{j=1}^n e_{\alpha_j})$ as a sum over stable trees $T$ from $\ST_{g,n+2}$. As we already explained in the proof of Theorem~\ref{theorem:finite DR hierarchy}, in formula~\eqref{eq:sum of trees for our family}, at each vertex $v$ of a stable tree $T$, the class $c_{g(v),n(v)}$ is fed by a covector $e^{\alpha_1(v)}$ and vectors $e_{\alpha_2(v)},\ldots,e_{\alpha_{n(v)}(v)}$, and the result is zero unless $\alpha_1(v)=\ldots=\alpha_{n(v)}(v)$. Moreover, in order to get a nontrivial result, we should necessarily have $\alpha(v)<\alpha(v')$ for any two vertices~$v$ and~$v'$ with $l_T(v)<l_T(v')$.

\medskip

\noindent\underline{\it Computation of the differential polynomials $P^2_{\beta,d}$}. In the computation of a differential polynomial $P^2_{\beta,d}$, at the root $v_0$ of $T$ the class $c_{g(v_0),n(v_0)}$ is fed by the covector $(\Id-R_1^t\psi_1)e^2=e^2$. If $\beta=1$, then at some vertex $v$ the class $c_{g(v),n(v)}$ is fed by the vector $(\Id-R_1\psi_2)e_1=e_1$, which implies that the result is zero. Therefore, $P^2_{1,d}=0$. If $\beta=2$, then at some vertex $v$ the class~$c_{g(v),n(v)}$ is fed by the vector $(\Id-R_1\psi_2)e_2=e_2-\xi\psi_2 e_1$. Feeding by $e_1$, as in the case $\beta=1$, immediately gives zero. Therefore, $c_{g(v),n(v)}$ must be fed by~$e_2$, which immediately implies that if the result is nonzero, then $T$ consists of just one vertex and, thus,
\begin{align*}
&\int_{\DR_g(-\sum_{j=1}^n a_j,0,a_1,\ldots,a_n)} \lambda_g\psi_2^d \left((\Id+R_1 z)c^{\triv,G}\right)_{g,n+2}(e^2\otimes e_2\otimes \otimes_{j=1}^n e_{\alpha_j})=\\
=&\int_{\DR_g(-\sum_{j=1}^n a_j,0,a_1,\ldots,a_n)} \lambda_g\psi_2^d c^{\triv,G}_{g,n+2}(e^2\otimes e_2\otimes \otimes_{j=1}^n e_{\alpha_j})=\\
=&(G^2)^g\int_{\DR_g(-\sum_{j=1}^n a_j,0,a_1,\ldots,a_n)} \lambda_g\psi_2^d.
\end{align*}
This shows that $P^2_{2,d}=\left.P_d^\KdV\right|_{u_n=u^2_n,\,\eps\mapsto\sqrt{G^2}\eps}$ and completes the proof of Part 2 of the proposition.
 
\medskip

\noindent\underline{\it Computation of the differential polynomial $P^1_{1,0}$}. Using arguments analogous to the ones from above, and also the degree counting argument from the proof of Theorem~\ref{theorem:finite DR hierarchy}, it is easy to see that the stable trees shown on Figure~\ref{figure:P11 graphs} are the only stable trees contributing to the integral~\eqref{eq:DR integral for N2} with $\alpha=\beta=1$ and $d=0$. On the figure, we put in a box the vector corresponding to the second marked point, which also corresponds to the point of multiplicity zero in the double ramification cycle. 
\begin{figure}[t]
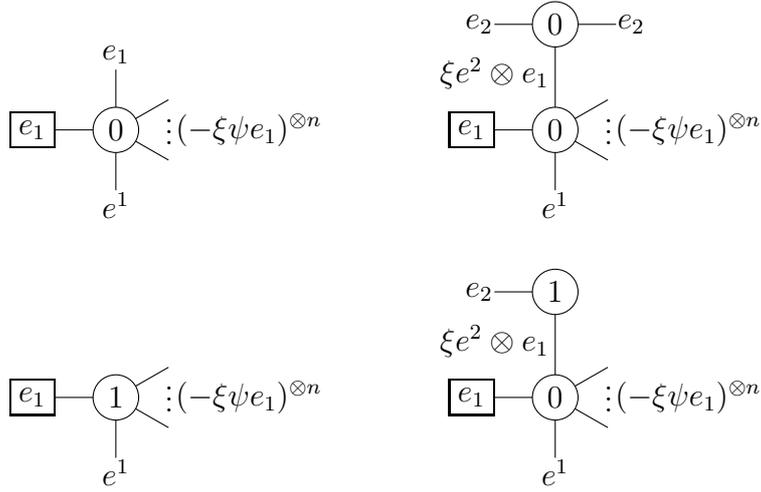

\tikz{
\leg{A}{30};\leg{A}{-30};\leg{A}{-90};\leg{A}{180};\leg{A}{90};\gg{0}{A};\lab{A}{0}{7mm}{\raisebox{0.3\height}{\vdots}};\lab{A}{0}{17.5mm}{\raisebox{0.3\height}{$(-\xi \psi e_1)^{\otimes n}$}};\lab{A}{-90}{10mm}{e^1};\lab{A}{90}{10mm}{e_1};\lab{A}{180}{11mm}{\boxed{e_1}};
}
\hspace{1cm}
\tikz{
\draw (A)--(B);\leg{A}{30};\leg{A}{-30};\leg{A}{-90};\leg{A}{180};\leg{B}{0};\leg{B}{180};\gg{0}{A};\gg{0}{B};\lab{A}{0}{7mm}{\raisebox{0.3\height}{\vdots}};\lab{A}{0}{17.5mm}{\raisebox{0.3\height}{$(-\xi \psi e_1)^{\otimes n}$}};\lab{A}{-90}{10mm}{e^1};\lab{B}{0}{10mm}{e_2};\lab{B}{180}{10mm}{e_2};\lab{A}{180}{11mm}{\boxed{e_1}};\node at (-8mm,7.5mm) {$\xi e^2\otimes e_1$};
}
\\[0.5cm]
\tikz{
\leg{A}{30};\leg{A}{-30};\leg{A}{-90};\leg{A}{180};\gg{1}{A};\lab{A}{0}{7mm}{\raisebox{0.3\height}{\vdots}};\lab{A}{0}{17.5mm}{\raisebox{0.3\height}{$(-\xi \psi e_1)^{\otimes n}$}};\lab{A}{-90}{10mm}{e^1};\lab{A}{180}{11mm}{\boxed{e_1}};
}
\hspace{1cm}
\tikz{
\draw (A)--(B);\leg{A}{30};\leg{A}{-30};\leg{A}{-90};\leg{A}{180};\leg{B}{180};\gg{0}{A};\gg{1}{B};\lab{A}{0}{7mm}{\raisebox{0.3\height}{\vdots}};\lab{A}{0}{17.5mm}{\raisebox{0.3\height}{$(-\xi \psi e_1)^{\otimes n}$}};\lab{A}{-90}{10mm}{e^1};\lab{B}{180}{10mm}{e_2};\lab{A}{180}{11mm}{\boxed{e_1}};\node at (-8mm,7.5mm) {$\xi e^2\otimes e_1$};
}
\caption{Stable trees contributing to $P^1_{1,0}$}
\label{figure:P11 graphs}
\end{figure}

\medskip

Consider the contribution of the trees of genus $0$. Note that $\int_{\oM_{0,n+3}}\psi_1\psi_2\cdots\psi_n=n!$. The first tree of genus $0$ on Figure~\ref{figure:P11 graphs} contributes to the integral
$$
\int_{\DR_0(-\sum_{j=1}^{n+1}a_j,0,a_1,\ldots,a_{n+1})}\left((\Id+R_1 z)c^{\triv,G}\right)_{0,n+3}(e^1\otimes e_1^{\otimes 2}\otimes e_2^{\otimes n})
$$
as $(-\xi)^n n!$. The second tree of genus $0$ on Figure~\ref{figure:P11 graphs} contributes to the integral
$$
\int_{\DR_0(-\sum_{j=1}^{n+2} a_j,0,a_1,\ldots,a_{n+2})}\left((\Id+R_1 z)c^{\triv,G}\right)_{0,n+3}(e^1\otimes e_1\otimes e_2^{\otimes (n+2)})
$$
as $(-1)^n\xi^{n+1}\frac{(n+2)!}{2}$. Therefore,
$$
\Coef_{\eps^0}P^1_{1,0}=\frac{u^1+\frac{\xi(u^2)^2}{2}}{1+\xi u^2}.
$$

\medskip

Consider now the contribution of the trees of genus $1$ from Figure~\ref{figure:P11 graphs}.  The first tree of genus~$1$ on Figure~\ref{figure:P11 graphs} contributes to the integral
$$
\int_{\DR_1(-\sum_{j=1}^n a_j,0,a_1,\ldots,a_n)}\lambda_1\left((\Id+R_1 z)c^{\triv,G}\right)_{0,n+2}(e^1\otimes e_1\otimes e_2^{\otimes n})
$$
as
$$
G^1(-\xi)^n\underbrace{\int_{\DR_1(-\sum_{j=1}^n a_j,0,a_1,\ldots,a_n)}\lambda_1\psi_3\cdots\psi_{n+2}}_{P_n(a_1,\ldots,a_n):=}.
$$

\medskip

Note that $P_n$ is a symmetric polynomial in $a_1,\ldots,a_n$ of degree $2$ satisfying
\begin{align*}
P_n|_{a_n=0}=&\int_{\DR_1(-\sum_{j=1}^{n-1} a_j,0,a_1,\ldots,a_{n-1},0)}\lambda_1\psi_3\cdots\psi_{n+2}=\\
=&(n+1)\int_{\DR_1(-\sum_{j=1}^{n-1} a_j,0,a_1,\ldots,a_{n-1})}\lambda_1\psi_3\cdots\psi_{n+1}=\\\
=&(n+1)P_{n-1}.
\end{align*}
Therefore,
$$
\left.\frac{P_n}{(n+1!)}\right|_{a_n=0}=\frac{P_{n-1}}{n!},
$$
which implies that there exist constants $\alpha$ and $\beta$ such that
$$
\frac{P_n}{(n+1!)}=\alpha m_{(2)}(a_1,\ldots,a_n)+\beta m_{(1,1)}(a_1,\ldots,a_n),
$$
where, for a partition $\lambda$, we denote by $m_\lambda(a_1,\ldots,a_n)$ the monomial symmetric function (see, e.g.,~\cite[Section~2]{Mac95}). In order to determine $\alpha$ and $\beta$, it is sufficient to compute $P_2$. We have
$$
P_2=\int_{\DR_1(-a_1-a_2,0,a_1,a_2)}\lambda_1\psi_3\psi_4=\int_{\DR_1(-a_1-a_2,a_1,a_2)}\lambda_1(\psi_2+\psi_3).
$$
Using that (see, e.g.,~\cite[Section 4.1]{BR16})
$$
\int_{\DR_1(-b_1-b_2,b_1,b_2)}\lambda_1\psi_1=\frac{b_1^2+b_2^2}{24},
$$
we obtain
$$
P_2=\frac{1}{24}(2(a_1+a_2)^2+a_1^2+a_2^2)=\frac{1}{24}(3(a_1^2+a_2^2)+4a_1a_2).
$$
Therefore, $\alpha=\frac{1}{48}$, $\beta=\frac{1}{36}$, and
\begin{gather}\label{eq:formula for Pn}
P_n=\frac{(n+1)!}{48}m_{(2)}+\frac{(n+1)!}{36}m_{(1,1)}.
\end{gather}

\medskip

Thus, the first tree of genus $1$ on Figure~\ref{figure:P11 graphs} gives the following contribution to the differential polynomial $P^1_{1,0}$:
\begin{align*}
&\eps^2 G^1\sum_{n\ge 1}(-\xi)^n\frac{n(n+1)}{48}u^2_{xx}(u^2)^{n-1}+\eps^2 G^1\sum_{n\ge 2}(-\xi)^n\frac{(n-1)n(n+1)}{72}(u^2_x)^2(u^2)^{n-2}=\\
=&-\eps^2\frac{\xi G^1}{24}\frac{u^2_{xx}}{(1+\xi u^2)^3}+\eps^2 \frac{\xi^2 G^1}{12}\frac{(u^2_x)^2}{(1+\xi u^2)^4}.
\end{align*}

\medskip

The second tree of genus~$1$ on Figure~\ref{figure:P11 graphs} contributes to the integral
$$
\int_{\DR_1(-\sum_{j=1}^{n+1} a_j,0,a_1,\ldots,a_{n+1})}\lambda_1\left((\Id+R_1 z)c^{\triv,G}\right)_{0,n+3}(e^1\otimes e_1\otimes e_2^{\otimes (n+1)})
$$
as
$$
G^2(-1)^n\xi^{n+1}n!\sum_{j=1}^{n+1}\int_{\DR_1(-a_j,a_j)}\lambda_1=\frac{G^2(-1)^n\xi^{n+1}n!}{24}m_{(2)}(a_1,\ldots,a_{n+1}),
$$
which gives the following contribution to the differential polynomial $P^1_{1,0}$:
$$
\eps^2\frac{\xi G^2}{24}\frac{u^2_{xx}}{1+\xi u^2}.
$$
This completes the computation of the differential polynomial $P^1_{1,0}$. It is clear that after the Miura transformation~\eqref{eq:special Miura transformation} we obtain the formulas for the flow~$\frac{\d}{\d t^1_0}$ in the variables~$\tu^1,\tu^2$ from the statement of the theorem.

\medskip

\noindent\underline{\it Computation of the differential polynomial $P^1_{2,0}$}. All the stable trees of genus $0$ contributing to the differential polynomial $P^1_{2,0}$ are shown on Figure~\ref{figure:P12 genus0 graphs},
\begin{figure}[t]
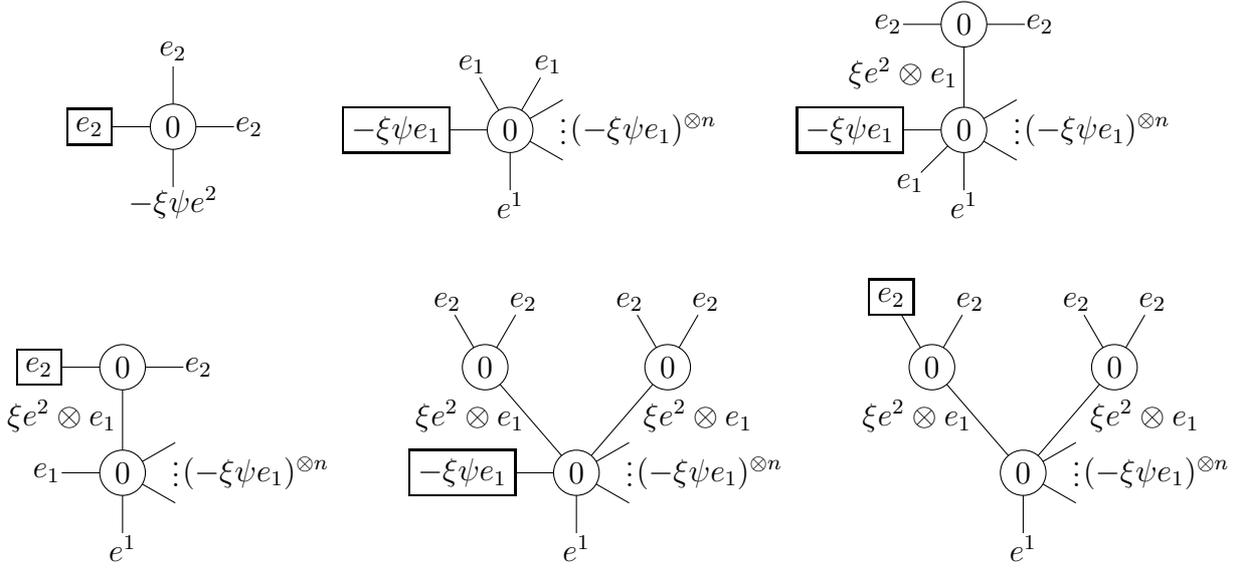

\tikz{
\leg{A}{0};\leg{A}{-90};\leg{A}{180};\leg{A}{90};\gg{0}{A};\lab{A}{0}{10mm}{e_2};\lab{A}{-90}{10mm}{-\xi\psi e^2};\lab{A}{90}{10mm}{e_2};\lab{A}{180}{11mm}{\boxed{e_2}};
}
\hspace{0.5cm}
\tikz{
\leg{A}{120};\leg{A}{60};\leg{A}{30};\leg{A}{-30};\leg{A}{-90};\leg{A}{180};\gg{0}{A};\lab{A}{0}{7mm}{\raisebox{0.3\height}{\vdots}};\lab{A}{0}{17.5mm}{\raisebox{0.3\height}{$(-\xi \psi e_1)^{\otimes n}$}};\lab{A}{-90}{10mm}{e^1};\lab{A}{180}{15mm}{\boxed{-\xi\psi e_1}};\lab{A}{60}{10mm}{e_1};\lab{A}{120}{10mm}{e_1};
}
\hspace{0.5cm}
\tikz{
\draw (A)--(B);\leg{A}{30};\leg{A}{-135};\leg{A}{-30};\leg{A}{-90};\leg{A}{180};\leg{B}{0};\leg{B}{180};\gg{0}{A};\gg{0}{B};\lab{A}{0}{7mm}{\raisebox{0.3\height}{\vdots}};\lab{A}{0}{17.5mm}{\raisebox{0.3\height}{$(-\xi \psi e_1)^{\otimes n}$}};\lab{A}{-90}{10mm}{e^1};\lab{B}{0}{10mm}{e_2};\lab{B}{180}{10mm}{e_2};\lab{A}{-135}{10mm}{e_1};\lab{A}{180}{15mm}{\boxed{-\xi\psi e_1}};\node at (-8mm,7.5mm) {$\xi e^2\otimes e_1$};
}
\\[0.5cm]
\tikz{
\draw (A)--(B);\leg{A}{30};\leg{A}{180};\leg{A}{-30};\leg{A}{-90};\leg{B}{0};\leg{B}{180};\gg{0}{A};\gg{0}{B};\lab{A}{0}{7mm}{\raisebox{0.3\height}{\vdots}};\lab{A}{0}{17.5mm}{\raisebox{0.3\height}{$(-\xi \psi e_1)^{\otimes n}$}};\lab{A}{-90}{10mm}{e^1};\lab{B}{0}{10mm}{e_2};\lab{A}{180}{10mm}{e_1};\lab{B}{180}{11mm}{\boxed{e_2}};\node at (-8mm,7.5mm) {$\xi e^2\otimes e_1$};
}
\hspace{0.5cm}
\tikz{
\draw (A)--(BL);\draw (A)--(BR);\leg{A}{30};\leg{A}{180};\leg{A}{-30};\leg{A}{-90};\leg{BL}{120};\leg{BL}{60};\leg{BR}{120};\leg{BR}{60};\gg{0}{A};\gg{0}{BL};\gg{0}{BR};\lab{A}{0}{7mm}{\raisebox{0.3\height}{\vdots}};\lab{A}{0}{17.5mm}{\raisebox{0.3\height}{$(-\xi \psi e_1)^{\otimes n}$}};\lab{A}{-90}{10mm}{e^1};\lab{BL}{60}{10mm}{e_2};\lab{BL}{120}{10mm}{e_2};\lab{BR}{60}{10mm}{e_2};\lab{BR}{120}{10mm}{e_2};\lab{A}{180}{15mm}{\boxed{-\xi\psi e_1}};\node at (-14mm,7.5mm) {$\xi e^2\otimes e_1$};\node at (16mm,7.5mm) {$\xi e^2\otimes e_1$};
}
\hspace{0.5cm}
\tikz{
\draw (A)--(BL);\draw (A)--(BR);\leg{A}{30};\leg{A}{-30};\leg{A}{-90};\leg{BL}{120};\leg{BL}{60};\leg{BR}{120};\leg{BR}{60};\gg{0}{A};\gg{0}{BL};\gg{0}{BR};\lab{A}{0}{7mm}{\raisebox{0.3\height}{\vdots}};\lab{A}{0}{17.5mm}{\raisebox{0.3\height}{$(-\xi \psi e_1)^{\otimes n}$}};\lab{A}{-90}{10mm}{e^1};\lab{BL}{60}{10mm}{e_2};\lab{BL}{120}{10.7mm}{\boxed{e_2}};\lab{BR}{60}{10mm}{e_2};\lab{BR}{120}{10mm}{e_2};\node at (-14mm,7.5mm) {$\xi e^2\otimes e_1$};\node at (16mm,7.5mm) {$\xi e^2\otimes e_1$};
}
\caption{Stable trees of genus $0$ contributing to $P^1_{2,0}$}
\label{figure:P12 genus0 graphs}
\end{figure}
and their contributions are the following:
\begin{align*}
& -\xi\frac{(u^2)^2}{2}, && -\xi\frac{(u^1)^2}{2}\frac{1}{(1+\xi u^2)^2} && -\xi^2\frac{u^1(u^2)^2}{2}\frac{1}{(1+\xi u^2)^2},\\
& \xi u^1 u^2\frac{1}{1+\xi u^2}, && -\xi^3\frac{(u^2)^4}{8}\frac{1}{(1+\xi u^2)^2}, && \xi^2\frac{(u^2)^3}{2}\frac{1}{1+\xi u^2}.
\end{align*}

\medskip

All the stable trees of genus $1$ contributing to the differential polynomial $P^1_{2,0}$ are shown on Figure~\ref{figure:P12 genus1 graphs}. 
\begin{figure}[t]
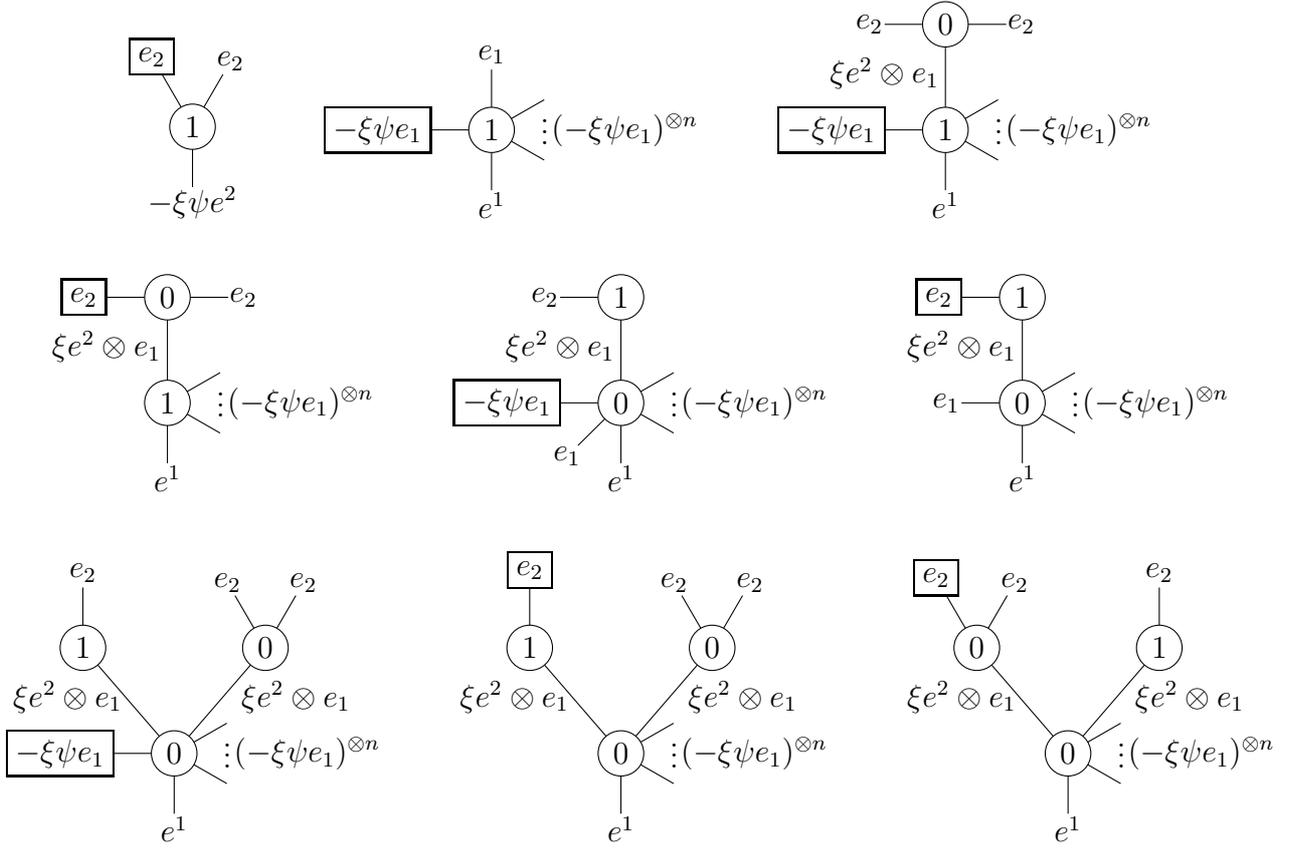

\tikz{
\leg{A}{-90};\leg{A}{120};\leg{A}{60};\gg{1}{A};\lab{A}{-90}{10mm}{-\xi\psi e^2};\lab{A}{60}{10mm}{e_2};\lab{A}{120}{10.7mm}{\boxed{e_2}};
}
\hspace{0.5cm}
\tikz{
\leg{A}{90};\leg{A}{30};\leg{A}{-30};\leg{A}{-90};\leg{A}{180};\gg{1}{A};\lab{A}{0}{7mm}{\raisebox{0.3\height}{\vdots}};\lab{A}{0}{17.5mm}{\raisebox{0.3\height}{$(-\xi \psi e_1)^{\otimes n}$}};\lab{A}{-90}{10mm}{e^1};\lab{A}{180}{15mm}{\boxed{-\xi\psi e_1}};\lab{A}{90}{10mm}{e_1};
}
\hspace{0.5cm}
\tikz{
\draw (A)--(B);\leg{A}{30};\leg{A}{-30};\leg{A}{-90};\leg{A}{180};\leg{B}{0};\leg{B}{180};\gg{1}{A};\gg{0}{B};\lab{A}{0}{7mm}{\raisebox{0.3\height}{\vdots}};\lab{A}{0}{17.5mm}{\raisebox{0.3\height}{$(-\xi \psi e_1)^{\otimes n}$}};\lab{A}{-90}{10mm}{e^1};\lab{B}{0}{10mm}{e_2};\lab{B}{180}{10mm}{e_2};\lab{A}{180}{15mm}{\boxed{-\xi\psi e_1}};\node at (-8mm,7.5mm) {$\xi e^2\otimes e_1$};
}
\\[0.5cm]
\tikz{
\draw (A)--(B);\leg{A}{30};\leg{A}{-30};\leg{A}{-90};\leg{B}{0};\leg{B}{180};\gg{1}{A};\gg{0}{B};\lab{A}{0}{7mm}{\raisebox{0.3\height}{\vdots}};\lab{A}{0}{17.5mm}{\raisebox{0.3\height}{$(-\xi \psi e_1)^{\otimes n}$}};\lab{A}{-90}{10mm}{e^1};\lab{B}{0}{10mm}{e_2};\lab{B}{180}{11mm}{\boxed{e_2}};\node at (-8mm,7.5mm) {$\xi e^2\otimes e_1$};
}
\hspace{0.5cm}
\tikz{
\draw (A)--(B);\leg{A}{30};\leg{A}{-135};\leg{A}{-30};\leg{A}{-90};\leg{A}{180};\leg{B}{180};\gg{0}{A};\gg{1}{B};\lab{A}{0}{7mm}{\raisebox{0.3\height}{\vdots}};\lab{A}{0}{17.5mm}{\raisebox{0.3\height}{$(-\xi \psi e_1)^{\otimes n}$}};\lab{A}{-90}{10mm}{e^1};\lab{B}{180}{10mm}{e_2};\lab{A}{-135}{10mm}{e_1};\lab{A}{180}{15mm}{\boxed{-\xi\psi e_1}};\node at (-8mm,7.5mm) {$\xi e^2\otimes e_1$};
}
\hspace{0.5cm}
\tikz{
\draw (A)--(B);\leg{A}{30};\leg{A}{180};\leg{A}{-30};\leg{A}{-90};\leg{B}{180};\gg{0}{A};\gg{1}{B};\lab{A}{0}{7mm}{\raisebox{0.3\height}{\vdots}};\lab{A}{0}{17.5mm}{\raisebox{0.3\height}{$(-\xi \psi e_1)^{\otimes n}$}};\lab{A}{-90}{10mm}{e^1};\lab{A}{180}{10mm}{e_1};\lab{B}{180}{11mm}{\boxed{e_2}};\node at (-8mm,7.5mm) {$\xi e^2\otimes e_1$};
}
\\[0.5cm]
\tikz{
\draw (A)--(BL);\draw (A)--(BR);\leg{A}{30};\leg{A}{180};\leg{A}{-30};\leg{A}{-90};\leg{BL}{90};\leg{BR}{120};\leg{BR}{60};\gg{0}{A};\gg{1}{BL};\gg{0}{BR};\lab{A}{0}{7mm}{\raisebox{0.3\height}{\vdots}};\lab{A}{0}{17.5mm}{\raisebox{0.3\height}{$(-\xi \psi e_1)^{\otimes n}$}};\lab{A}{-90}{10mm}{e^1};\lab{BL}{90}{10mm}{e_2};\lab{BR}{60}{10mm}{e_2};\lab{BR}{120}{10mm}{e_2};\lab{A}{180}{15mm}{\boxed{-\xi\psi e_1}};\node at (-14mm,7.5mm) {$\xi e^2\otimes e_1$};\node at (16mm,7.5mm) {$\xi e^2\otimes e_1$};
}
\hspace{0.5cm}
\tikz{
\draw (A)--(BL);\draw (A)--(BR);\leg{A}{30};\leg{A}{-30};\leg{A}{-90};\leg{BL}{90};\leg{BR}{120};\leg{BR}{60};\gg{0}{A};\gg{1}{BL};\gg{0}{BR};\lab{A}{0}{7mm}{\raisebox{0.3\height}{\vdots}};\lab{A}{0}{17.5mm}{\raisebox{0.3\height}{$(-\xi \psi e_1)^{\otimes n}$}};\lab{A}{-90}{10mm}{e^1};\lab{BL}{90}{10.3mm}{\boxed{e_2}};\lab{BR}{60}{10mm}{e_2};\lab{BR}{120}{10mm}{e_2};\node at (-14mm,7.5mm) {$\xi e^2\otimes e_1$};\node at (16mm,7.5mm) {$\xi e^2\otimes e_1$};
}
\hspace{0.5cm}
\tikz{
\draw (A)--(BL);\draw (A)--(BR);\leg{A}{30};\leg{A}{-30};\leg{A}{-90};\leg{BL}{120};\leg{BL}{60};\leg{BR}{90};\gg{0}{A};\gg{0}{BL};\gg{1}{BR};\lab{A}{0}{7mm}{\raisebox{0.3\height}{\vdots}};\lab{A}{0}{17.5mm}{\raisebox{0.3\height}{$(-\xi \psi e_1)^{\otimes n}$}};\lab{A}{-90}{10mm}{e^1};\lab{BL}{60}{10mm}{e_2};\lab{BL}{120}{10.7mm}{\boxed{e_2}};\lab{BR}{90}{10mm}{e_2};\node at (-14mm,7.5mm) {$\xi e^2\otimes e_1$};\node at (16mm,7.5mm) {$\xi e^2\otimes e_1$};
}
\caption{Stable trees of genus $1$ contributing to $P^1_{2,0}$.}
\label{figure:P12 genus1 graphs}
\end{figure}
In order to compute the contributions of stable trees number 2, 3, and 4, we have to compute the integral
\begin{gather}\label{eq:genus 1 integral}
\int_{\DR_1(-\sum_{j=1}^{n+1} a_j,a_{n+1},a_1,\ldots,a_n)}\lambda_1\psi_3\cdots\psi_{n+2}.
\end{gather}
When $a_{n+1}=0$, we have already done that above, see equation~\eqref{eq:formula for Pn}. A computation in the case of arbitrary $a_{n+1}$ is analogous, and we obtain
$$
\int_{\DR_1(-\sum_{j=1}^{n+1} a_j,a_{n+1},a_1,\ldots,a_n)}\lambda_1\psi_3\cdots\psi_{n+2}=(n+1)!\left(\frac{1}{48}m_{(2)}+\frac{1}{36}m_{(1,1)}+\frac{a_{n+1}}{24}m_{(1)}+\frac{a_{n+1}^2}{24}\right),
$$
where $m_\lambda=m_\lambda(a_1,\ldots,a_n)$. As a result, the contributions of the stable trees on Figure~\ref{figure:P12 genus1 graphs} to $\Coef_{\eps^2}P^1_{2,0}$ are the following: 
\begin{align*}
\text{\small tree number 1:}\hspace{0.3cm}&-\frac{\xi G^2}{24}u^2_{xx} \\
\text{\small trees number 2 and 3:}\hspace{0.3cm}&\frac{\xi^2 G^1}{8}\frac{\left(u^1+\frac{\xi(u^2)^2}{2}\right) u^2_{xx}}{(1+\xi u^2)^4}-\frac{\xi^3 G^1}{3}\frac{\left(u^1+\frac{\xi(u^2)^2}{2}\right)(u^2_x)^2}{(1+\xi u^2)^5}\\
&+\frac{\xi^2 G^1}{4}\frac{\left(u^1+\frac{\xi(u^2)^2}{2}\right)_x u^2_x}{(1+\xi u^2)^4}-\frac{\xi G^1}{12}\frac{\left(u^1+\frac{\xi(u^2)^2}{2}\right)_{xx}}{(1+\xi u^2)^3}\\
\text{\small tree number 4:}\hspace{0.3cm}&\frac{\xi G^1}{24}\frac{u^2_{xx}}{(1+\xi u^2)^3}-\frac{\xi^2 G^1}{12}\frac{(u^2_x)^2}{(1+\xi u^2)^4}\\
\text{\small trees number 5 and 7:}\hspace{0.3cm}& -\frac{\xi^2 G^2}{24}\frac{\left(u^1+\frac{\xi(u^2)^2}{2}\right)u^2_{xx}}{(1+\xi u^2)^2}\\
\text{\small trees number 6 and 8:}\hspace{0.3cm}& 0 \\
\text{\small tree number 9:}\hspace{0.3cm}& \frac{\xi^2 G^2}{24}\frac{u^2 u^2_{xx}}{1+\xi u^2}
\end{align*}
where for the trees number 6 and 8 we get zero, because $\DR_g(0,\ldots,0)=(-1)^g\lambda_g$ and $\lambda_g^2=0$ for $g\ge 1$.

\medskip

All the stable trees of genus $2$ contributing to the differential polynomial $P^1_{2,0}$ are shown on Figure~\ref{figure:P12 genus2 graphs},
\begin{figure}[t]
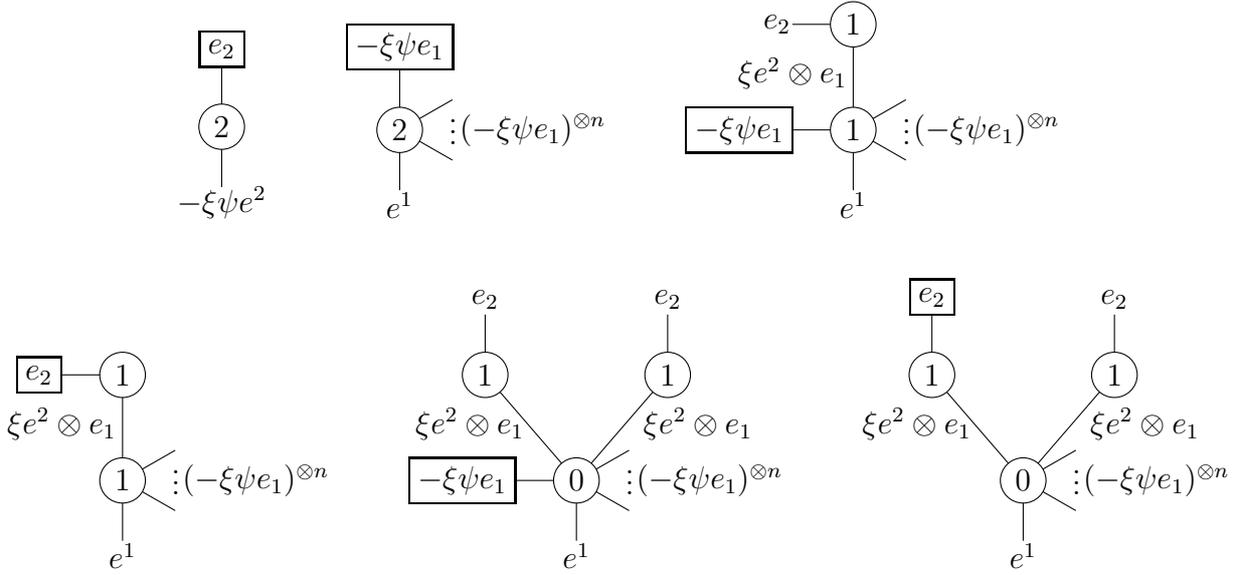

\tikz{
\leg{A}{-90};\leg{A}{90};\gg{2}{A};\lab{A}{-90}{10mm}{-\xi\psi e^2};\lab{A}{90}{10.3mm}{\boxed{e_2}};
}
\hspace{0.5cm}
\tikz{
\leg{A}{30};\leg{A}{-30};\leg{A}{-90};\leg{A}{90};\gg{2}{A};\lab{A}{0}{7mm}{\raisebox{0.3\height}{\vdots}};\lab{A}{0}{17.5mm}{\raisebox{0.3\height}{$(-\xi \psi e_1)^{\otimes n}$}};\lab{A}{-90}{10mm}{e^1};\lab{A}{90}{11mm}{\boxed{-\xi\psi e_1}};
}
\hspace{0.5cm}
\tikz{
\draw (A)--(B);\leg{A}{30};\leg{A}{-30};\leg{A}{-90};\leg{A}{180};\leg{B}{180};\gg{1}{A};\gg{1}{B};\lab{A}{0}{7mm}{\raisebox{0.3\height}{\vdots}};\lab{A}{0}{17.5mm}{\raisebox{0.3\height}{$(-\xi \psi e_1)^{\otimes n}$}};\lab{A}{-90}{10mm}{e^1};\lab{B}{180}{10mm}{e_2};\lab{A}{180}{15mm}{\boxed{-\xi\psi e_1}};\node at (-8mm,7.5mm) {$\xi e^2\otimes e_1$};
}\\[0.5cm]
\tikz{
\draw (A)--(B);\leg{A}{30};\leg{A}{-30};\leg{A}{-90};\leg{B}{180};\gg{1}{A};\gg{1}{B};\lab{A}{0}{7mm}{\raisebox{0.3\height}{\vdots}};\lab{A}{0}{17.5mm}{\raisebox{0.3\height}{$(-\xi \psi e_1)^{\otimes n}$}};\lab{A}{-90}{10mm}{e^1};\lab{B}{180}{11mm}{\boxed{e_2}};\node at (-8mm,7.5mm) {$\xi e^2\otimes e_1$};
}
\hspace{0.5cm}
\tikz{
\draw (A)--(BL);\draw (A)--(BR);\leg{A}{30};\leg{A}{180};\leg{A}{-30};\leg{A}{-90};\leg{BL}{90};\leg{BR}{90};\gg{0}{A};\gg{1}{BL};\gg{1}{BR};\lab{A}{0}{7mm}{\raisebox{0.3\height}{\vdots}};\lab{A}{0}{17.5mm}{\raisebox{0.3\height}{$(-\xi \psi e_1)^{\otimes n}$}};\lab{A}{-90}{10mm}{e^1};\lab{BL}{90}{10mm}{e_2};\lab{BR}{90}{10mm}{e_2};\lab{A}{180}{15mm}{\boxed{-\xi\psi e_1}};\node at (-14mm,7.5mm) {$\xi e^2\otimes e_1$};\node at (16mm,7.5mm) {$\xi e^2\otimes e_1$};
}
\hspace{0.5cm}
\tikz{
\draw (A)--(BL);\draw (A)--(BR);\leg{A}{30};\leg{A}{-30};\leg{A}{-90};\leg{BL}{90};\leg{BR}{90};\gg{0}{A};\gg{1}{BL};\gg{1}{BR};\lab{A}{0}{7mm}{\raisebox{0.3\height}{\vdots}};\lab{A}{0}{17.5mm}{\raisebox{0.3\height}{$(-\xi \psi e_1)^{\otimes n}$}};\lab{A}{-90}{10mm}{e^1};\lab{BL}{90}{10.3mm}{\boxed{e_2}};\lab{BR}{90}{10mm}{e_2};\node at (-14mm,7.5mm) {$\xi e^2\otimes e_1$};\node at (16mm,7.5mm) {$\xi e^2\otimes e_1$};
}
\caption{Stable trees of genus $2$ contributing to $P^1_{2,0}$.}
\label{figure:P12 genus2 graphs}
\end{figure}
and their contributions to $\Coef_{\eps^4}P^1_{2,0}$ are the following:
\begin{align*}
\text{\small tree number 1:}\hspace{0.3cm}& 0 \\
\text{\small tree number 2:}\hspace{0.3cm}& (G^1)^2\left(\frac{\xi^2}{288}\frac{u^2_{xxxx}}{(1+\xi u^2)^5}-\frac{29 \xi^3}{1152}\frac{(u^2_{xx})^2}{(1+\xi u^2)^6}-\frac{11\xi^3}{288}\frac{u^2_{xxx}u^2_x}{(1+\xi u^2)^6}\right.\\
& \hspace{1.2cm}\left.+\frac{5\xi^4}{24}\frac{u^2_{xx}(u^2_x)^2}{(1+\xi u^2)^7}-\frac{49\xi^5}{288}\frac{(u^2_x)^4}{(1+\xi u^2)^8}\right) \\
\text{\small tree number 3:}\hspace{0.3cm}& G^1G^2\left(\frac{\xi^3}{192}\frac{(u^2_{xx})^2}{(1+\xi u^2)^4}-\frac{\xi^4}{72}\frac{u^2_{xx}(u^2_x)^2}{(1+\xi u^2)^5}+\frac{\xi^3}{96}\frac{u^2_{xxx}u^2_x}{(1+\xi u^2)^4}-\frac{\xi^2}{288}\frac{u^2_{xxxx}}{(1+\xi u^2)^3}\right) \\
\text{\small tree number 4:}\hspace{0.3cm}& 0 \\
\text{\small tree number 5:}\hspace{0.3cm}& -\frac{\xi^3(G^2)^2}{1152}\frac{(u^2_{xx})^2}{(1+\xi u^2)^2} \\
\text{\small tree number 6:}\hspace{0.3cm}& 0 
\end{align*}
where the formula for the contribution of the tree number 2 is based on the computation of the integral
\begin{gather}\label{eq:genus 2 integral}
\int_{\DR_2(-\sum_{i=1}^na_i,a_1,\ldots,a_n)}\hspace{-1cm}\lambda_2\psi_2\cdots\psi_{n+1}=(n+2)!\left(\frac{m_{(4)}}{6912}+\frac{29 m_{(2,2)}}{69120}+\frac{11 m_{(3,1)}}{34560}+\frac{m_{(2,1,1)}}{1728}+\frac{7 m_{(1,1,1,1)}}{8640}\right),
\end{gather}
with $m_\lambda=m_\lambda(a_1,\ldots,a_n)$. This computation is similar to the computation of the integral~\eqref{eq:genus 1 integral}. The integral~\eqref{eq:genus 2 integral} is a homogeneous polynomial in $a_1,\ldots,a_n$ of degree $4$, which we denote by $Q_n(a_1,\ldots,a_n)$. Then using the dilaton equation we observe that
$$
\frac{Q_n(a_1,\ldots,a_{n-1},0)}{n+2}=Q_{n-1}(a_1,\ldots,a_{n-1}),
$$
which implies that 
$$
\frac{Q_n(a_1,\ldots,a_n)}{(n+2)!}=\alpha m_{(4)}+\beta m_{(2,2)}+\gamma m_{(3,1)}+\delta m_{(2,1,1)}+\zeta m_{(1,1,1,1)},
$$
for some constants $\alpha,\beta,\gamma,\delta,\zeta$ that do not depend on $n$. In order to determine them, it is enough to compute the polynomial $Q_4$: we did it using the formula for the intersection of a psi-class with the DR cycle from~\cite[Theorem~4]{BSSZ15} and with the help of \textsf{Mathematica}.

\medskip

Collecting all the computed contributions to the differential polynomial $P^1_{2,0}$ and doing the Miura transformation~\eqref{eq:special Miura transformation}, again with the help of \textsf{Mathematica}, we obtain the equations for the flow $\frac{\d}{\d t^2_0}$ in the variables $\tu^1,\tu^2$ from the statement of the theorem.
\end{proof}

\end{document}